\documentclass[letterpaper, 10 pt, conference]{ieeeconf}  

\IEEEoverridecommandlockouts                              
\usepackage{amsmath} 
\usepackage{amssymb}  

\usepackage{bbm}

\usepackage{mathtools}
\usepackage{subcaption}
\usepackage[ruled,vlined]{algorithm2e}
\usepackage{url}
\usepackage{float}

\usepackage{graphicx}
\usepackage{caption}
\usepackage{booktabs}
\usepackage{varwidth}
\usepackage{multirow}
\usepackage{cite}

\usepackage{caption}
\usepackage{subcaption}

\newtheorem{theorem}{Theorem}
\newtheorem{corollary}{Corollary}
\newtheorem{lemma}{Lemma}

\title{\LARGE \bf
Robust Online and Distributed Mean Estimation \\Under Adversarial Data Corruption}

\author{Tong Yao and Shreyas Sundaram
\thanks{Tong Yao and Shreyas Sundaram are with Elmore Family School of Electrical and Computer Engineering at Purdue University. 
Email: \{yao127,  sundara2\}@purdue.edu
}  
\thanks{This research was supported by National Science Foundation (NSF) grant CMMI 1638311.
}
}

\begin{document}

\maketitle
\thispagestyle{empty}
\pagestyle{empty}


\begin{abstract}
We study robust mean estimation in an online and distributed scenario in the presence of adversarial data attacks. At each time step, each agent in a network receives a potentially corrupted data point, where the data points were originally independent and identically distributed samples of a random variable. We propose online and distributed algorithms for all agents to asymptotically estimate the mean. We provide the error-bound and the convergence properties of the estimates to the true mean under our algorithms. Based on the network topology, we further evaluate each agent's trade-off in convergence rate between incorporating data from neighbors and learning with only local observations.
\end{abstract}


\section{Introduction} 
Multi-agent cooperative learning scenarios, where agents in a network collect data and coordinate with each other to make inferences, have received significant attention in the research community. In the era of large-scale and decentralized systems, distributed learning has many applications in machine learning, data science, and decision making, e.g., multi-armed bandits \cite{landgren2016distributed}, distributed correlation estimation \cite{DGAMA}, hypothesis testing \cite{mitra2020new}, etc. 

Due to the distributed nature of such systems, the data gathered by each agent may be corrupted in various ways (e.g., through adversarial attacks, or faulty sensor readings). Learning the true values and reaching consensus robustly and resiliently in the presence of corrupted data and misbehaving agents have been studied extensively in the literature. For example, \cite{mitra2020new} and \cite{yin2018byzantine} propose distributed learning algorithms that are robust against Byzantine attacks, \cite{dubey2020cooperative} considers distributed heavy-tail stochastic bandit problems with robust mean estimators, and \cite{leblanc2013resilient} analyzes resilient consensus in the presence of malicious agents.

Estimating the expected value is one of the most fundamental problems in statistics. Robust statistics address the problem of mean estimation under outliers and attacks, with classical estimators including median-of-means \cite{ALON1999137} and the trimmed mean \cite{tukey1963trim,ronchetti2009robust}. We refer the readers to the survey on robust mean estimation \cite{lugosi2019survey}. Early pioneering works on robust statistics attempt to estimate the mean given an outlier model \cite{tukey1960survey, huber1992robust}, and recent advancements consider stronger contamination models \cite{diakonikolas2020outlier, lugosi2021robust}. 

In this paper, we study the problem of estimating the mean of a stream of corrupted data arriving at agents in the network.  We design online and distributed algorithms to achieve robust estimation of the mean; these algorithms allow each agent to utilize its local data stream and information from neighbors to obtain accurate estimates of the true mean with high confidence, given a sequence of arbitrarily corrupted data under a strong contamination model.  

\subsubsection*{Contributions} First, we provide two online algorithms that recursively utilize a modified trimmed mean estimator to estimate the mean from arbitrarily corrupted data (with an upper bound on the fraction of corrupted samples). Each of the two algorithms has different computational requirements and provides different performance guarantees. Second, we extend our online algorithm to a cooperative multi-agent setting, such that all the agents in the network can collaboratively estimate the mean (given a certain amount of arbitrarily corrupted data), and improve their estimates through communication with their neighbors.  In all cases, we provide bounds on the estimation errors as a function of the number of samples and the fraction of corrupted samples.  

\section{Problem Formulation}
\subsection{Data and Corruption Model}
Let $X \in \mathbb{R}$ be a real-valued random variable that has finite variance $\sigma_X^2$. Let $\mu = \mathbb{E}[X]$ and define $\overline{X} = X - \mu$. The mean is assumed to be unknown, and we wish to estimate it through our algorithms. In this work, we assume that $X$ has an absolutely continuous distribution.
At each time step $t \in \{1,2,\ldots\}$, an independent and identically distributed copy of $X$ is collected. We denote the data point collected at time $t$ by $x_t$ and denote the set of data collected up to time $t$ as $X_t = \{x_1,x_2, \ldots, x_t\}.$ 

We consider the existence of an adversary who can inspect all the samples and replace them with arbitrary values. We assume that due to a limited budget of the adversary, the total number of corrupted samples is at most $\eta t$ up to time $t$, where the corruption parameter $\eta \in (0, 1/16)$. 
We will discuss the range of the corruption in Section \ref{sec: background}. We say that a set of samples is $\eta$-corrupted if it is generated by the above process. We apply the tilde notation to any potentially corrupted data or dataset, e.g., $\tilde{x}_t$ is a potentially corrupted data and $\tilde{X}_t = \{\tilde{x}_1,\tilde{x}_2,\ldots,\tilde{x}_t\}$ is a set of $\eta$-corrupted data up to time $t$.

\subsection{Network Model}
Consider a group of $m$ agents, $m \in \mathbb{N}$, with a known, unweighted, undirected, and connected communication graph $\mathcal{G} = (\mathcal{V},\mathcal{E})$, where $\mathcal{V} = \{1,2,\ldots,m\}$ is the set of vertices representing the agents and $\mathcal{E} \subseteq \mathcal{V} \times \mathcal{V}$ is the set of edges. If an edge $({i},{j}) \in \mathcal{E}$, agent $i$ and $j$ can communicate with each other. The neighbors of agent $i \in \mathcal{V}$ and agent $i$ itself are represented by the set $\mathcal{N}_i = \{j \in \mathcal{V}: (i,j) \in \mathcal{E}\} \cup \{i\}$, which is the inclusive neighborhood of agent $i$. 

At each time step $t$, each agent $i \in \mathcal{V}$ collects a potentially corrupted data point $\tilde{x}_{i,t}$. In this work, we consider a scenario where at each time step $t$, every agent in the network $i \in \mathcal{V}$ has at most $\eta t$ corrupted data. 

This paper aims to first design \textit{online} algorithms such that given any $\eta-$corrupted data arriving in a stream, the algorithms efficiently output good estimates of the mean of $X$ with high probability. Secondly, we formulate a \textit{distributed} algorithm for each agent $i$ to calculate and update the estimates of the mean of data in real-time. 
Given a sequence of corrupted observations of agent $i$, denoted by $\{\tilde{x}_{i,1},\tilde{x}_{i,2},\tilde{x}_{i,3},\ldots\}$, the objective of $i$ is to perform an online inference of the mean $\hat{\mu}_{i,t}$ at time $t$ by incorporating its local data and information from its neighbors. In particular, we want all agents to reach consensus on their estimates asymptotically, i.e., as $t \to \infty$, the estimate of each agent converges asymptotically to the agreement: 
$\lim_{t \to \infty}\hat{\mu}_{1,t} = \lim_{t \to \infty}\hat{\mu}_{2,t} = \ldots = \lim_{t \to \infty}\hat{\mu}_{m,t} \approx \mu.$  

\section{Trimmed Mean Estimator}\label{sec: background}
In this section, we introduce a modified version of the trimmed mean estimator in \cite{lugosi2021robust} on which we build our online and distributed algorithms. To simplify notations, we describe the estimator for any agent $i \in \mathcal{V}$ and omit the subscript $i$.

At time step $2t$, each agent has $2t$ samples after an adversary corrupts up to $2 \eta t$ points from the original independent and identically distributed (i.i.d.) samples of the random variable $X$. The samples are arbitrarily divided into two sets, denoted by $\tilde{X}_t = \{\tilde{x}_{1},\ldots,\tilde{x}_{t}\}$ and $\tilde{Y}_t = \{\tilde{y}_{1}, \ldots, \tilde{y}_{t}\}$. 
Given the corruption parameter $\eta \in (0,1/16)$ and a desired confidence level $\delta \in (0,1)$, define 
\begin{equation}\label{eqn: epsilon def}
\epsilon_t = 8\eta + 12\frac{\log(4/\delta)}{t}.    
\end{equation}
Note that for sufficiently large $t$, $\epsilon_t < 0.5$ since $\eta < 1/16$. 

Let $\tilde{y}_1^* \leq \tilde{y}_2^* \leq \ldots \leq \tilde{y}_t^*$ be a non-decreasing rearrangement of $\tilde{y}_1, \ldots, \tilde{y}_t \in \mathbb{R}$. Define $\alpha_t = \tilde{y}_{\epsilon_t t}^*$ and $\beta_t = \tilde{y}_{(1-\epsilon_t)t}^*$ to be {\it trimming values}. Note that $\alpha$ and $\beta$ are both agent specific and time specific but we omit the agent subscript $i$ and time subscript $t$ temporarily for convenience.
For $\alpha \leq \beta$, and $x \in \mathbb{R}$, define the trim estimator  
\begin{equation}
\phi_{\alpha,\beta}(x) = 
\begin{cases}
\beta &\text{if } x > \beta,\\
x &\text{if } x \in [\alpha,\beta],\\
\alpha &\text{if } x < \alpha.
\end{cases}    
\end{equation}
The estimated mean after trimming is given by 
\[\hat{\mu}_{t} = \frac{1}{t} \sum_{j = 1}^t \phi_{\alpha,\beta}(\tilde{x}_{j}). \]

For $0<p<1$, define the quantile $Q_p(\overline{X}) = \sup\{M \in \mathbb{R}: \mathbb{P}(\overline{X} \geq M) \geq 1-p\}$. We have $\mathbb{P}(\overline{X} \geq Q_p(\overline{X})) = 1-p$ and from Chebyshev’s inequality, 
\begin{equation*}
\mathbb{P}(\overline{X} \geq Q_p(\overline{X})) = 1-p \leq \frac{\sigma_X^2}{ Q_p^2(\overline{X})}   
\end{equation*}
and as a result,
\begin{equation}\label{eqn: upper bound on Q}
|Q_{p}(\overline{X})| \leq \frac{\sigma_X}{\sqrt{1-p}}.
\end{equation}
The estimation error is defined in \cite{lugosi2021robust} to be 
\begin{equation}\label{eqn: E}
E(\epsilon_t,X) \coloneqq \max\{\mathbb{E}[|\overline{X}|_{\mathbf{1}_{\overline{X} \leq Q_{\epsilon_t /2}(\overline{X})}}],\mathbb{E}[|\overline{X}|_{\mathbf{1}_{ \overline{X} \geq Q_{1-\epsilon_t/2}(\overline{X})}}]\},   
\end{equation} where $\mathbf{1}$ is the indicator function and it is equal to one when the condition is satisfied, and zero otherwise. 

From \cite{lugosi2021robust}, we have that for every $X$,
\begin{equation}\label{eqn: upper of E}
    E(\epsilon_t,X) \leq \sigma_X\sqrt{8\epsilon_t}.
\end{equation}

The following result upper bounds the trimming values using quantiles. 

\begin{lemma}[\cite{lugosi2021robust}]\label{lemma: Q bounds}
Consider the corruption-free sample $y_1,\ldots,y_t$. With probability at least $1-4e^{-\epsilon_tt/12}$, the inequalities
\begin{align*}
    |\{i: y_i \geq \mu + Q_{1-2\epsilon_t}(\overline{X}\}| &\geq 3/2\epsilon_tt\\
    |\{i: y_i \leq \mu + Q_{1-\epsilon_t/2}(\overline{X}\}| &\geq (1-(3/4)\epsilon_t)t\\
    |\{i: y_i \leq \mu + Q_{2\epsilon_t}(\overline{X}\}| &\geq 3/2\epsilon_tt\\
    |\{i: y_i \geq \mu + Q_{\epsilon_t/2}(\overline{X}\}| &\geq (1-(3/4)\epsilon_t)t
\end{align*} hold simultaneously. We denote the event when the above four inequalities hold as event $A$. On event $A$, since corruption $\eta \leq \epsilon_t/8$, the following inequalities also hold 
\begin{equation}\label{eqn: beta}
    Q_{1-2\epsilon_t}(\overline{X}) \leq \beta_t - \mu \leq Q_{1-\epsilon_t/2}(\overline{X}),
\end{equation}
\begin{equation} \label{eqn: alpha}
    Q_{\epsilon_t/2}(\overline{X}) \leq \alpha_t - \mu\leq Q_{2\epsilon_t}(\overline{X}).
\end{equation} 
\end{lemma}

The result below provides insight on the quality of the estimated mean at time $t$.
\begin{lemma}[\cite{lugosi2021robust}]\label{lemma: orig}
Let $\delta \in (0,1)$ and $\delta \geq 4e^{-t}$. Using the trimmed mean estimator, with probability at least $1-\delta$,
\begin{equation}\label{eqn: batch bound}
|\hat{\mu}_{t} - \mu| \leq 3 E(4\epsilon_t,X) + 2\sigma_X\sqrt{\frac{\log(4/\delta)}{t}}.
\end{equation}
\end{lemma}


\section{Online Robust Mean Estimation} \label{sec: online}
In the previous section, we provided an overview of the trimmed mean estimator for batch data.  In this section, we provide update algorithms for each agent to estimate the mean from a {\it sequence of data arriving in real-time}. 

\subsection{Fixed Trimming Thresholds}
We assume that each agent starts with a set of potentially $\eta$-corrupted data containing all the data available to the agent at time $t_0 \geq 1$, denoted as $\tilde{Y}_{i,t_0} = \tilde{X}_{i,t_0} =  \{\tilde{x}_{i,1}, \ldots, \tilde{x}_{i,t_0}\}$. At time $t = t_0$, each agent estimates the trimming thresholds $\alpha_{i,t_0}$ and $\beta_{i,t_0}$ following the procedures described in Section \ref{sec: background}. For all $t > t_0$, each agent starts the estimation process recursively using the previous estimates, 
\begin{equation}\label{eqn: mu local update 1}
    \hat{\mu}_{i,t} = \frac{(t-t_0-1)\hat{\mu}_{i,t-1} + \phi_{\alpha_{i,t_0},\beta_{i,t_0}}(\tilde{x}_{i,t})}{t - t_0}.
\end{equation}
We include the pseudo-code implementation for the online update in Algorithm \ref{algo: online mean 1}. 
\begin{algorithm}[ht]
\SetKwInOut{Initialization}{Initialization}
\SetKwInOut{Parameter}{Parameter}
\Initialization{$\hat{\mu}_{i,t_0} = 0$}
\KwIn{$t_0, \delta, \eta$, a sequence of corrupted data $\tilde{x}_{i,t}$}
\For{$t \leq t_0$}{
Store data in $\tilde{Y}_{i,t}$\\
\If{$t = t_0$}{
Compute $\epsilon_{t_0}$ as per \eqref{eqn: epsilon def}\\
Sort data in $\tilde{Y}_{i,t_0}$ and determine trimming thresholds $\alpha_{i,t_0}, \beta_{i,t_0}$}}
\For{$t > t_0$}{
Update mean estimate as     
$\hat{\mu}_{i,t} = \frac{(t-t_0-1)\hat{\mu}_{i,t-1} + \phi_{\alpha_{i,t_0},\beta_{i,t_0}}(\tilde{x}_{i,t})}{t-t_0}$
}
\caption{Online Robust Mean Estimation - Fixed Trimming Thresholds}
\label{algo: online mean 1}
\end{algorithm}

We provide the following results that guarantee the quality of the estimates. 
\begin{theorem}\label{thm: online 1}
Let $\delta \in (0,1)$ be s.t. $\delta \geq 4e^{-(t-t_0)}, \forall t \geq 2t_0$. Following the procedures of Algorithm \ref{algo: online mean 1}, with probability at least $1-\delta$, the estimates of each agent $i\in\mathcal{V}$ satisfy
\begin{equation}\label{eqn: mu bound algo 1}
|\hat{\mu}_{i,t}  - \mu| \leq 3E(4\epsilon_{t_0},X) + 2\sigma_X\sqrt{\frac{\log(4/\delta)}{t-t_0}} ,\forall t \geq 2t_0.  
\end{equation}
\end{theorem}
The proof of Theorem \ref{thm: online 1} is similar to Lemma \ref{lemma: orig} (batch algorithm). We include the proof in the Appendix for completeness. 

Using the definition of $\epsilon_t$ in \eqref{eqn: epsilon def} and property \eqref{eqn: upper of E} with Theorem \ref{thm: online 1}, we obtain the result below on the convergence of the estimates.
\begin{corollary}
Let $\delta \in (0,1)$ be s.t. $\delta \geq 4e^{-(t-t_0)}, \forall t \geq 2t_0$. Following the procedures of Algorithm \ref{algo: online mean 1}, for any $\xi >0$, there is a sufficiently large $\bar{t}$, s.t. $\forall t \geq \bar{t}$, the estimates of each agent $i\in\mathcal{V}$ satisfy the following inequality with probability at least $1-\delta$, 
\begin{equation}
|\hat{\mu}_{i,t}  - \mu| \leq 24\sigma_X\sqrt{4\eta + 6\frac{\log(4/\delta)}{t_0}} + \xi.   
\end{equation}
\end{corollary}

From the above results, we observe that the estimation error converges to the error introduced by initialization, i.e., $3E(4\epsilon_{t_0},X)$. Thus, it requires a relatively large $t_0$ to give reasonable estimates.
However, Algorithm \ref{algo: online mean 1} is suitable when the memory of agent is limited, as the algorithm only requires $O(1)$ memory. If each agent has sufficient memory, maintaining a sorted list of data $\tilde{Y}_t$ to update the trimming thresholds $\alpha_t$ and $\beta_t$ as new data arrives has advantages since the initial errors will converge to zero as data continuously stream in. 

\subsection{Updated Trimming Thresholds}
In this subsection, we propose an alternative online algorithm that updates the trimming thresholds $\alpha_t$ and $\beta_t$ as new data point arrives, improving the performance of the trimming operator while sacrificing $O(t)$ memory and $O(\log t)$ computation. 
In comparison, the batch algorithm \cite{lugosi2021robust} considers static $\alpha$ and $\beta$ with a given dataset and the proposed online Algorithm \ref{algo: online mean 1} still uses static trimming values determined at $t_0$ but updates the mean recursively as data stream in. 

We first describe the online algorithm and subsequently analyze the convergence properties and the trade-offs between the online and the batch algorithm. 

At time step $2t$, the samples of each agent are arbitrarily divided into two sets, denoted $\tilde{X}_{i,t} = \{\tilde{x}_{i,1},\ldots,\tilde{x}_{i,t}\}$ and $\tilde{Y}_{i,t} = \{\tilde{y}_{i,1}, \ldots, \tilde{y}_{i,t}\}$. In order to simplify the notation and discussion, we assume two data points $\tilde{x}_{i,t}$ and $\tilde{y}_{i,t}$ arrive at each time step $t$.

Each agent will store $\tilde{Y}_{i,t}$ as a binary search tree (BST), denoted BST($\tilde{Y}_{i,t}$). 
At each time step $t$, each agent updates $\epsilon_{t}$ following \eqref{eqn: epsilon def} and updates the trimming thresholds $\alpha_{i,t} = \tilde{y}_{i,\epsilon_t t}^*$ and $\beta_{i,t} = \tilde{y}_{i,(1-\epsilon_t)t}^*$ in $O(\ln t)$ complexity with the BST data structure. 
Subsequently, each agent performs a recursive local update using its previous estimates $\hat{\mu}_{i,t-1}$ and applying the new trimming thresholds $\alpha_{i,t},\beta_{i,t}$ on data $\tilde{x}_{i,t}$:
\begin{equation}\label{eqn: mu local update}
    \hat{\mu}_{i,t} = \frac{(t-1)\hat{\mu}_{i,t-1} + \phi_{\alpha_{i,t},\beta_{i,t}}(\tilde{x}_{i,t})}{t}.
\end{equation}
The pseudo-code is provided in Algorithm \ref{algo: online mean 2} below.
\begin{algorithm}[ht]
\SetKwInOut{Initialization}{Initialization}
\SetKwInOut{Parameter}{Parameter}
\Initialization{$\hat{\mu}_{i,0} = 0$}
\KwIn{$\delta$, $\eta$, a sequence of corrupted data $\tilde{x}_{i,t}$ and $\tilde{y}_{i,t}$}
\BlankLine
\For{$t \in \mathbb{N}$}{
Update BST($\tilde{Y}_{i,t}$) with $\tilde{y}_{i,t}$\\
Compute $\epsilon_t$ as per \eqref{eqn: epsilon def}\\ Update $\alpha_{i,t}$ and $\beta_{i,t}$\\
$\hat{\mu}_{i,t} = \frac{(t-1)\hat{\mu}_{i,t-1} + \phi_{\alpha_{i,t},\beta_{i,t}}(\tilde{x}_{i,t})}{t}$
}
\caption{Online Robust Mean Estimation - Updated Trimming Thresholds}
\label{algo: online mean 2}
\end{algorithm}

We provide the following results that guarantee the quality of the estimates of Algorithm \ref{algo: online mean 2}. 

\begin{theorem} \label{thm: online 2}
Following the procedures of Algorithm \ref{algo: online mean 2}, for all sample paths in a set of measure 1, there exists a finite time $\bar{t}$, such that for all $t \geq \bar{t}$, the error between the estimated mean and the true mean of any agent $i\in\mathcal{V}$, denoted as $|\hat{\mu}_{i,t} - \mu|$, satisfies the error bound
\begin{equation}\label{eqn: algo 2 error bound}
    |\hat{\mu}_{i,t} - \mu| \leq \frac{\bar{t}}{t} |\hat{\mu}_{i,\bar{t}} -\mu| + \frac{1}{t}\sum_{j = \bar{t}+1}^t\frac{\sigma_X}{\sqrt{4\eta + 6\frac{\log(4/\delta)}{j}}}.
\end{equation}
\end{theorem}
\begin{proof}
To simplify notation, we omit the subscript $i$ denoting the agent in the proof. Let $\phi_t = \phi_{\alpha_t,\beta_t}$ denote the updated trimming operator at time step $t \in \mathbb{N}$. 

Recall that the inequalities in Lemma \ref{lemma: Q bounds} hold simultaneously with probability at least $1-4\exp\{-\epsilon_t t/12\}, \forall t \in \mathbb{N}$.
Define a bad event at time $t$ to be the event that any inequalities in \eqref{eqn: beta} or \eqref{eqn: alpha} do not hold. Define the random variable $B_t$, with $B_t = 1$ if the bad event occurs at the given time $t$, and 0 otherwise. 
Let $B = \sum_{j = 1}^t B_j$ be the number of bad events up to time $t$. Summing up the probability of bad events, we have 
\begin{equation}
    \sum_{j = 1}^t \mathbb{P}[B_j] \leq 4\sum_{j = 1}^t e^{-\epsilon_j j/12}.
\end{equation}
As $t \to \infty$, it can be shown that the series is convergent, i.e., $\sum_{j = 1}^{\infty} \mathbb{P}[B_j] < \infty$. From the Borel-Cantelli lemma, the probability of infinitely many bad events occurring is 0. Thus, for a set of sample paths of measure 1, there exists a sample path dependent finite time $\Bar{t}$ such that no more bad events occur for $t \geq \bar{t}$. 

Rewriting the update for Algorithm \ref{algo: online mean 2} with finite time $\bar{t}$, the deviation between the estimated mean and the true mean can be expressed as
\begin{multline}\label{eqn: algo2 intermediate}
|\hat{\mu}_t - \mu| = |\frac{1}{t}(\sum_{j = 1}^{\bar{t}} \phi_{j}(\tilde{x}_j) + \sum_{j = \bar{t}+1}^t \phi_{j}(\tilde{x}_j)) - \mu| \\
\leq |\frac{1}{t} \sum_{j = 1}^{\bar{t}} \phi_{j}(\tilde{x}_j) - \frac{\bar{t}}{t} \mu| + |\frac{1}{t}\sum_{j = \bar{t}+1}^t \phi_{j}(\tilde{x}_j)- \frac{t-\bar{t}}{t}\mu |.  
\end{multline}

Since \eqref{eqn: beta} holds for all $t\geq \bar{t}$, from \eqref{eqn: upper bound on Q} and \eqref{eqn: epsilon def}, we have
$
\phi_t(\tilde{x}_t) \leq \beta_t \leq Q_{1-\epsilon_t/2} + \mu 
\leq \frac{\sigma_X \sqrt{2}}{\sqrt{\epsilon_t}} + \mu = \frac{\sigma_X}{\sqrt{4\eta + 6\frac{\log(4/\delta)}{t}}} + \mu.
$
Let $U_t = \frac{\sigma_X}{\sqrt{4\eta + 6\frac{\log(4/\delta)}{t}}}$. We can see that $U_{t-1}\leq U_t \leq \frac{\sigma_X}{\sqrt{4\eta}}$.
Using $U_t$, we can provide an upper bound
\begin{align*}
\frac{1}{t}\sum_{j = \bar{t}+1}^t \phi_{j}(\tilde{x}_j) \leq \frac{1}{t}\sum_{j = \bar{t}+1}^t (U_j + \mu) = \frac{1}{t}\sum_{j = \bar{t}+1}^t U_j + \frac{t-\bar{t}}{t}\mu.
\end{align*} Similarly with the lower tail, 
$
\phi_t(\tilde{x}_t) \geq \alpha_t \geq Q_{\epsilon_t/2} + \mu 
\geq -\frac{\sigma_X }{\sqrt{1 - \epsilon_t/2}} + \mu \geq - \frac{\sigma_X \sqrt{2}}{\sqrt{\epsilon_t}} + \mu = -U_t + \mu.
$
Thus, we have
$
\frac{1}{t}\sum_{j = \bar{t}+1}^t \phi_{j}(\tilde{x}_j) \geq -\frac{1}{t}\sum_{j = \bar{t}+1}^t U_j + \frac{t-\bar{t}}{t}\mu.    
$

\end{proof}
The above result provides a closed-form upper bound on the quality of the estimates from Algorithms \ref{algo: online mean 2}. We can observe from the below corollary, that deviation of the estimate and the true mean is indeed bounded as $t \to \infty$.
\begin{corollary}\label{coro: algo 2 conv}
Following the procedures of Algorithm \ref{algo: online mean 2}, for all sample paths in a set of measure 1, as $t\to\infty$, the estimates of each agent satisfy
\begin{equation}\label{eqn: mu bound algo 2}
\limsup_{t\to\infty}|\hat{\mu}_{i,t}  - \mu| \leq \frac{\sigma_X}{2\sqrt{\eta}}.    
\end{equation}
\end{corollary}

In order to make inferences in real-time, additional costs in accuracy occur from Algorithm \ref{algo: online mean 2}. As observed from \eqref{eqn: algo 2 error bound}, the error between the estimates of Algorithm \ref{algo: online mean 2} and the true mean has a term $\frac{\bar{t}}{t}|\hat{\mu}_{i,t} - \mu|$, introduced by the accumulated error from trimming when $t$ is relatively small. 
However, from Corollary \ref{coro: algo 2 conv}, we see that asymptotically, 
the estimates are only influenced by the corruption rate $\eta$. With the sacrifice of $O(t)$ space and additional $O(\ln t)$ computation complexity, Algorithm \ref{algo: online mean 2} eliminates the initialization error.

\section{Distributed Robust Mean Estimation} \label{sec: distributed}
Previously, we provide algorithms that allow each agent in the network to update their estimates of the mean in an online manner, using only local data. In this section, we propose a multi-agent distributed algorithm that enables the agents to estimate the mean {\it collaboratively} (i.e., by exchanging information with each other). First, we describe the proposed algorithm. Subsequently, we analyze the convergence properties and demonstrate the improvement in the learning rate.

To simplify analysis, an agent $o \in \mathcal{V}$ is randomly selected initially, whose trimming values $\alpha_o$ and $\beta_o$ are eventually transmitted to all agents in the network, i.e., $\alpha_i = \alpha_o$ and $\beta_i = \beta_o, \forall i \in \mathcal{V}$. Depending on the algorithm used, $\alpha_o$ and $\beta_o$ can be fixed or updated at each time step. 

At the beginning of each time step $t$, upon receiving a new data point $\tilde{x}_{i,t}$, agent $i$ updates its local estimate $\hat{\mu}_{i,t}$ using either Algorithm \ref{algo: online mean 1} or Algorithm \ref{algo: online mean 2}, determined by the system designer according to hardware conditions and applications.  
After the local update, each agent transmits its estimated mean $\hat{\mu}_{i,t} = \hat{\mu}_{i,t}^0$ to its neighbors. Here, we use superscript $k$ in $\hat{\mu}_{i,t}^k$ to denote the rounds of communications. At each time step, each agent can communicate $K \in \mathbb{N}$ times with its neighbors, updating its local estimates following
\begin{equation}\label{eqn: mu global update}
    \hat{\mu}_{i,t}^{k+1} = \frac{\sum_{j\in\mathcal{N}_i}\hat{\mu}_{j,t}^k}{|\mathcal{N}_i|}, \quad k = 0,\ldots,K-1.
\end{equation}

We provide the pseudo-code in Algorithm \ref{algo: distMean}. We only include the online local update from Algorithm \ref{algo: online mean 1} as the distributed algorithm can be modified for Algorithm \ref{algo: online mean 2}.

\begin{algorithm}
\SetKwInOut{Initialization}{Initialization}
\SetKwInOut{Parameter}{Parameter}
\Initialization{$\hat{\mu}_{i,t_0}^K = 0$}
\KwIn{$t_0, \delta, \eta, K$, a sequence of corrupted data $\tilde{x}_{i,t}$}
\For{$t \in \{t_0, t_0 + 1, \ldots\}$}{
\If{$t < t_0$}{
Store data in $\tilde{Y}_{i,t}$}
\If{$t = t_0$}{
\If{$i = o$}{
Compute $\epsilon_{t_0}$ as per \eqref{eqn: epsilon def}\\
Sort data in $\tilde{Y}_{i,t_0}$, determine and transmit trimming thresholds $\alpha_{o,t_0}, \beta_{o,t_0}$}}
\Else{Receive and transmit $\alpha_{o,t_0}, \beta_{o,t_0}$}
\If{$t > t_0$}{
Update mean as     
\begin{equation}\label{eqn: MA update}
 \hat{\mu}_{i,t}^0 = \frac{(t-t_0-1)\hat{\mu}_{i,t-1}^K + \phi_{\alpha_{o,t_0},\beta_{o,t_0}}(\tilde{x}_{i,t})}{t-t_0}   
\end{equation}
\For{$k = 1:K$}{
Transmit $\hat{\mu}_{i,t}^{k-1}$ to neighbors\\
Receive $\hat{\mu}_{\mathcal{N}_i\setminus\{i\},t}^{k-1}$ from neighbors\\
Update mean estimate $\hat{\mu}_{i,t}^{k} = \frac{\sum_{j\in\mathcal{N}_i}\hat{\mu}_{j,t}^{k-1}}{|\mathcal{N}_i|}$
}
}
}
\caption{Distributed Robust Mean Estimation of Agent $i$}
\label{algo: distMean}
\end{algorithm}

We analyze the theoretical properties of the proposed distributed algorithm. First, we rewrite the update \eqref{eqn: mu global update} as
\begin{equation}\label{eqn: mu global matrix}
    \hat{\mu}_{t}^{k+1} = (I + D)^{-1}(I+A) \hat{\mu}_{t}^k = P\hat{\mu}_{t}^k,
\end{equation} where $\hat{\mu}_{t}^{k} \in \mathbb{R}^{m}$ is a vector of all agents' estimates at communication round $k$ of time step $t$, $I\in\mathbb{R}^{m\times m}$ is the identity matrix, $D$ is the degree matrix, $A$ is the adjacency matrix of the communication network, and finally, $P = (I + D)^{-1}(I+A)$ is referred to as the Perron matrix \cite{olfati2007consensus}.   
At each fixed time step $t$, the group decision value for all agents, denoted $\Tilde{\mu}_t$, is the average of all initial mean estimates $\hat{\mu}_{i,t}^0$ at time $t$. It can be shown that for Algorithm \ref{algo: online mean 1}, the consensus value at time $t$ is
\begin{equation}\label{eqn: consensus 1}
    \tilde{\mu}_{t} = \frac{1}{m(t-t_0)} \sum_{i = 1}^m \sum_{j = t_0+1}^t \phi_{\alpha_{o,t_0},\beta_{o,t_0}}(\tilde{x}_{i,j}).
\end{equation}
If all agents apply Algorithm \ref{algo: online mean 2}, the consensus value at time $t$ can be expressed as
\begin{equation}\label{eqn: consensus 2}
    \tilde{\mu}_{t} = \frac{1}{mt} \sum_{i = 1}^m \sum_{j = 1}^t \phi_{\alpha_{o,j},\beta_{o,j}}(\tilde{x}_{i,j}).
\end{equation}

We provide the following lemmas that bound the error of the group decision mean value to the true mean. The proofs are included in the Appendix. 
\begin{lemma}\label{lemma: consensus conv 1}
Let $\delta \in (0,1)$ be s.t. $\delta \geq 4e^{-(t-t_0)}, \forall t \geq 2t_0$. Following Algorithm \ref{algo: online mean 1}, the consensus value of the entire network $\Tilde{\mu}_t$ satisfies the below condition with probability at least $1-\delta$,
\begin{equation}\label{eqn: tilde mu bound 1}
|\Tilde{\mu}_t  - \mu| \leq 3E(4\epsilon_{t_0},X) + 2\sigma_X\sqrt{\frac{\log(4/\delta)}{m(t-t_0)}}, \forall t \geq 2t_0.
\end{equation}
\end{lemma}
From the results above, we see that the convergence rate of the final estimates of each agent to the true mean is increased by a factor of $\sqrt{1/m}$. In Algorithm \ref{algo: online mean 2}, since the upper bound is loose, we do not observe a rate improvement from the following result.
\begin{lemma}\label{lemma: consensus conv 2}
Following Algorithm \ref{algo: online mean 2}, for all sample paths in a set of measure 1, there exists a finite time $\bar{t}$, such that for all $t \geq \bar{t}$, the consensus value of the entire network $\Tilde{\mu}_t$ satisfies condition
\begin{equation}\label{eqn: tilde mu bound 2}
|\Tilde{\mu}_t  - \mu| \leq \frac{\bar{t}}{t} |\tilde{\mu}_{\bar{t}} - \mu| + \frac{1}{t}\sum_{j = \bar{t}+1}^t\frac{\sigma_X}{\sqrt{4\eta + 6\frac{\log(4/\delta)}{j}}}.
\end{equation}
\end{lemma}


In order to analyze the convergence properties of the multi-agent algorithm, we will require the classical result below.
\begin{lemma}[\cite{olfati2007consensus}]\label{lemma: consensus}
Following \eqref{eqn: mu global matrix}, a discrete consensus is globally exponentially reached with a speed that is faster or equal to $\lambda$ for a connected undirected network, where $0 < \lambda < 1$ is the second largest eigenvalue of the matrix $P$. The error of the estimates is given by
\begin{equation}
|\hat{\mu}_{i,t}^k - \tilde{\mu}_t| \leq c\lambda^k |\hat{\mu}_{i,t}^0 - \tilde{\mu}_t|,    
\end{equation} where $c$ is some constant.
\end{lemma}

Below, we present the bound of error for the proposed distributed algorithms and include the proof in the Appendix. 
\begin{theorem}\label{thm: dist convergence 1}
For all $t \geq 2t_0, \forall k \in \{1,2,\ldots,K\}$, and $\forall i \in\mathcal{V}$, let $\delta \in (0,1)$ be s.t. $\delta \geq 4e^{-(t-t_0)}$. When agents follow Algorithm \ref{algo: online mean 1}, the estimated mean of each agent satisfies the following condition with probability at least $1-\delta$,
\begin{multline}\label{eqn: bound of conv algo 1}
|\hat{\mu}_{i,t}^K - \mu| \leq (c\lambda^{K})^{t-t_0}|\hat{\mu}_{i,t_0}^K - \tilde{\mu}_{t_0}| 
\\+ \sum_{j = t_0+1}^t (c\lambda^{K})^{t+1-j}  \frac{2\max(|\alpha_{o,t_0}|,|\beta_{o,t_0}|)}{j-t_0}\\
+ 3E(4\epsilon_{t_0},X) + 2\sigma_X\sqrt{\frac{\log(4/\delta)}{m(t-t_0)}}.    
\end{multline}
\end{theorem}

\begin{corollary}
Let $\delta \in (0,1)$ be s.t. $\delta \geq 4e^{-(t-t_0)}, \forall t \geq 2t_0$.
For all $k \in \{1,2,\ldots,K\}$ and $\forall K \in \mathbb{N}$ s.t. $c\lambda^K <1$, for any $\xi > 0$, there is a sufficiently large $\bar{t}$, s.t. $\forall t \geq \bar{t}$, when agents $i\in\mathcal{V}$ follow Algorithm \ref{algo: online mean 1}, the estimated mean of each agent satisfies the following condition with probability at least $1-\delta$,
\begin{equation}
|\hat{\mu}_{i,t}^K - \mu|\leq 24\sigma_X\sqrt{4\eta + 6\frac{\log(4/\delta)}{t_0}} + \xi. 
\end{equation}
\end{corollary}
\begin{proof}
From \eqref{eqn: bound of conv algo 1}, as $t \to \infty$, the first term goes to zero. The second term is an instance of Lemma \ref{lemma: nedic} (see Appendix), where $\zeta = c\lambda^K < 1$ and $\gamma_t = \frac{\max(|\alpha_{o,t_0}|,|\beta_{o,t_0}|)}{t-t_0} \to 0.$ As $t \to \infty$, the fourth term also converges to zero. 
\end{proof}

Using the previous results and technique of proofs, we arrive at the convergence properties if all agents follow Algorithm \ref{algo: online mean 2}.
\begin{theorem}
For all sample paths in a set of measure 1, there exists a finite time $\bar{t}$, such that $\forall t \geq \bar{t} + 1, \forall k \in \{1,2,\ldots,K\}$, if all agents $i\in\mathcal{V}$ follow Algorithm \ref{algo: online mean 2}, the estimated mean of each agent satisfies 
\begin{multline}\label{eqn: bound of conv algo 2}
|\mu_{i,t}^K - \mu| \leq (c\lambda^{K})^{t-\bar{t}}|\hat{\mu}_{i,\bar{t}}^K - \tilde{\mu}_{\bar{t}}| 
\\+ \sum_{j = \bar{t}+1}^t (c\lambda^{K})^{t+1-j} \frac{2\max(|\alpha_{o,j}|,|\beta_{o,j}|)}{j}
\\+ \frac{\bar{t}}{t} |\tilde{\mu}_{\bar{t}} - \mu| + \frac{1}{t}\sum_{j = \bar{t}+1}^t\frac{\sigma_X}{\sqrt{4\eta + 6\frac{\log(4/\delta)}{j}}}.
\end{multline}
\end{theorem}
\begin{proof}
For all $t \geq \bar{t} + 1$, we have
$
|e_t| \leq \frac{2\max(|\alpha_{o,t}|,|\beta_{o,t}|)}{t}, 
$ and $|e_t| \to 0$ as $t\to \infty$.
The remaining proof is similar to Theorem \ref{thm: dist convergence 1}.
\end{proof}
\begin{corollary}
For all $k \in \{1,2,\ldots,K\}$, and $\forall K \in \mathbb{N}$ such that $c\lambda^K <1$, if all agents $i\in\mathcal{V}$ follow Algorithm \ref{algo: online mean 2}, for all samples paths in a set of measure 1, as $t \to \infty$, the estimated mean of each agent satisfies 
\begin{equation}
\limsup_{t\to\infty}|\hat{\mu}_{i,t}^K - \mu| \leq \frac{\sigma_X}{2\sqrt{\eta}}.  
\end{equation}
\end{corollary}

From the above results, the asymptotic rate at which $\hat{\mu}_{i,t}^K$ converges to the true mean is dominated by the convergence of the data and the amount of corruption. Also, compared to the centralized algorithm, increasing the total rounds of communication $K$ of each time step will decrease the deviation of each agent's estimates to the centralized estimates. As time $t \to \infty$, the estimation errors converge to that of the centralized algorithms, with a rate improvement of $\sqrt{1/m}$ for Algorithm \ref{algo: online mean 1}. 

\section{Simulations}
In this section, we demonstrate the performance of our proposed algorithms via simulations. 
We generate data points from a Gaussian distribution $\mathcal{G}_1$ with $\mu_1 = 0, \sigma^2_1 = 1$. We then let the dataset to be $\eta = 0.02$ corrupted with data sampled from another Gaussian distribution $\mathcal{G}_2$ where $\mu_2 = 100$ and $\sigma^2_2 = 1$. We let $\delta = 0.3$. 

In Fig. \ref{fig: algo 1 and 2}, we plot the estimated mean of Algorithm \ref{algo: online mean 1} and \ref{algo: online mean 2} over the number of data points used for the estimation ($t_0 = 100$). Algorithm \ref{algo: online mean 2} obtains a better estimate asymptotically. The estimated mean of Algorithm \ref{algo: online mean 1} is $0.095$ and the estimated mean of Algorithm \ref{algo: online mean 2} is $-0.016$, given $1000$ data points for estimations. Both algorithms achieve good results as without the trimmed operator, the estimated mean is $1.98$. 

We observe that Algorithm \ref{algo: online mean 1} converges to the error introduced by initialization while with additional computation and memory complexity, Algorithm \ref{algo: online mean 2} obtains better estimates. 
\begin{figure}[H]
    \centering
    \includegraphics[width = 0.54\linewidth]{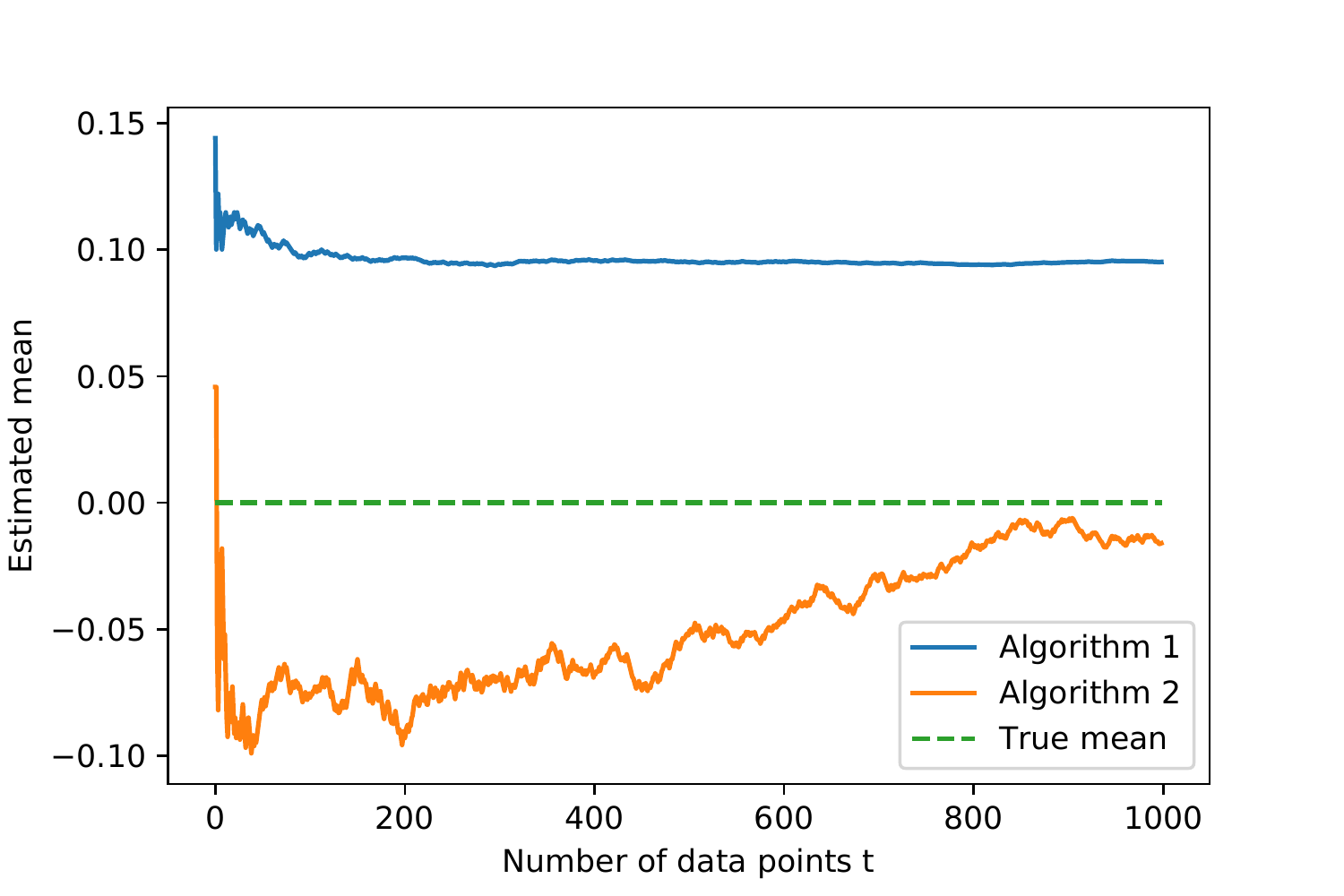}
    \caption{Estimated mean of Algorithm \ref{algo: online mean 1} and \ref{algo: online mean 2}}
    \label{fig: algo 1 and 2}
\end{figure}
\begin{figure}[H]
    \centering
    \begin{subfigure}[b]{0.49\linewidth}
         \centering
         \includegraphics[width=\textwidth]{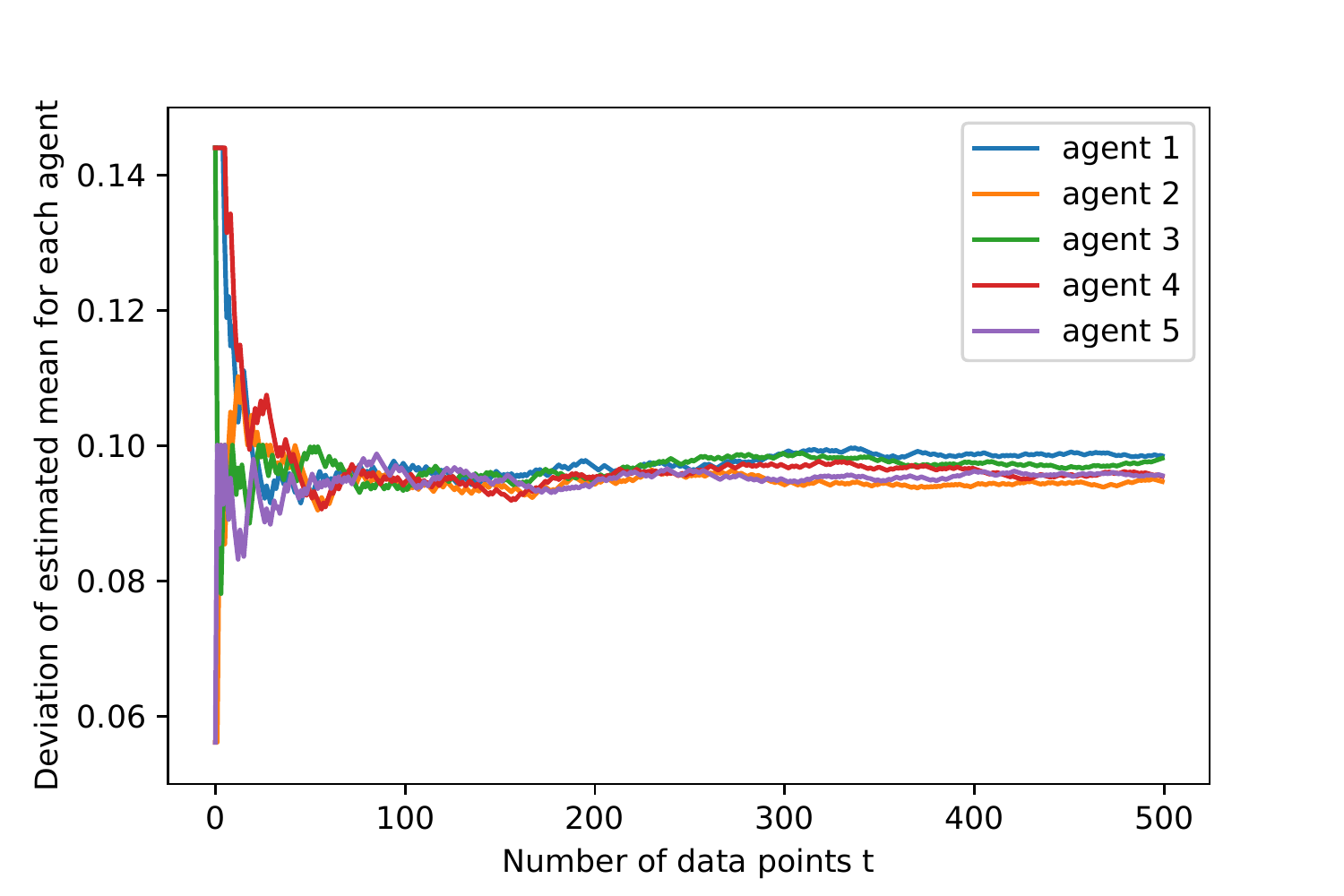}
         \caption{$K = 0$}
    \end{subfigure}
    \begin{subfigure}[b]{0.49\linewidth}
         \centering
         \includegraphics[width=\textwidth]{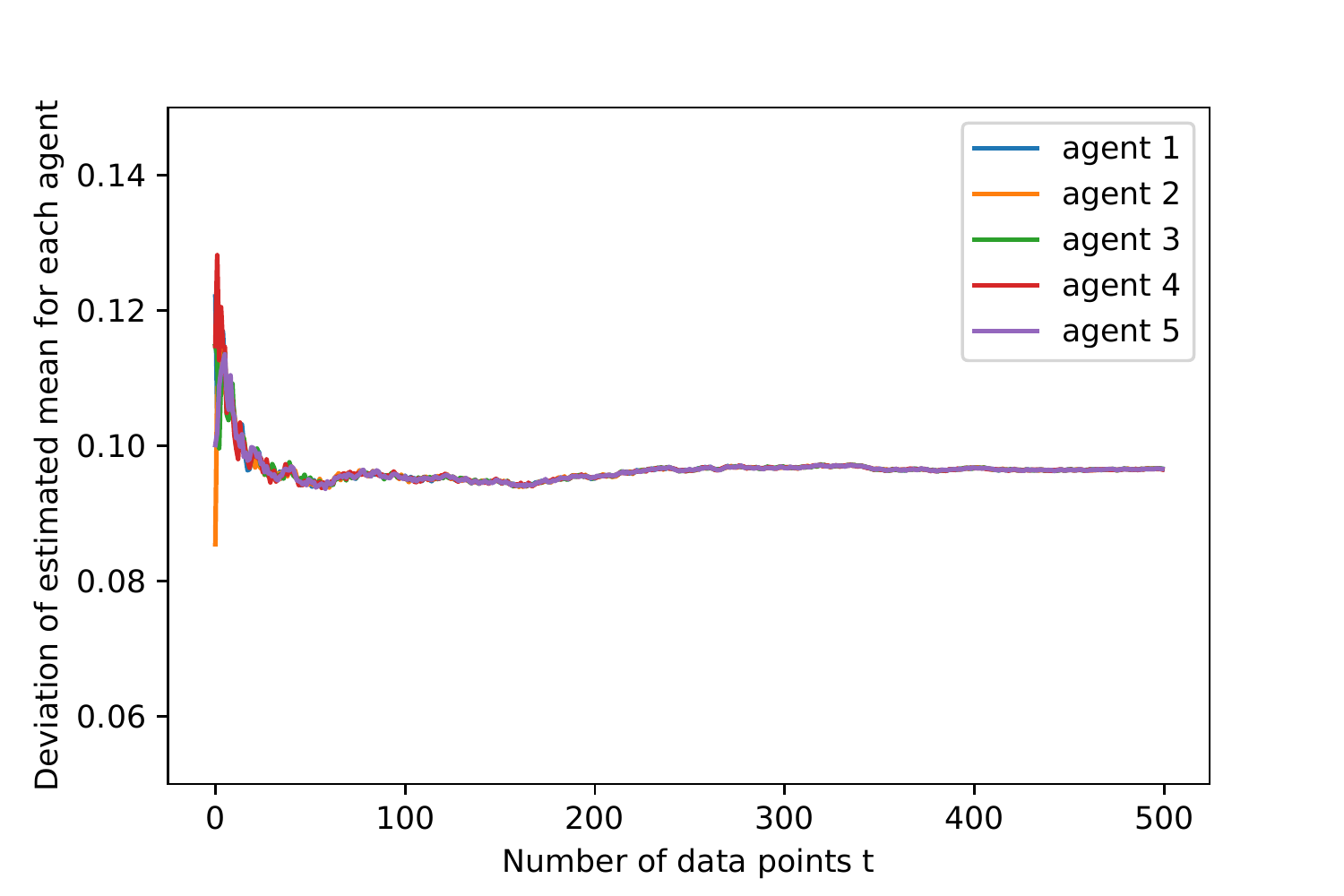}
         \caption{$K = 1$}
     \end{subfigure}
     \caption{Algorithm \ref{algo: online mean 1} with different parameter $K$.}
     \label{fig: dist simu 1}
\end{figure}
\begin{figure}[H]
    \centering
    \begin{subfigure}[b]{0.49\linewidth}
         \centering
         \includegraphics[width=\textwidth]{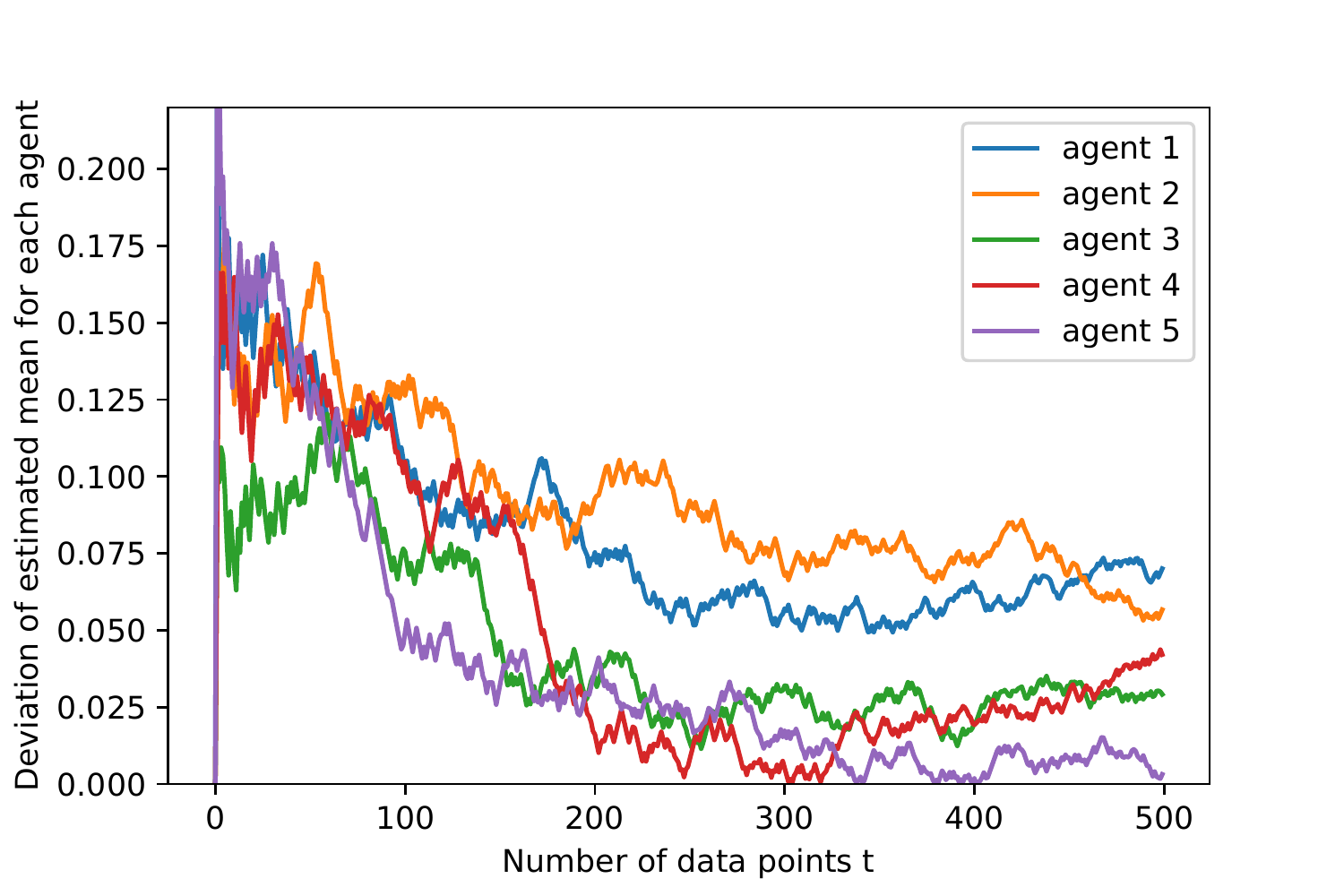}
         \caption{$K = 0$}
    \end{subfigure}
    \begin{subfigure}[b]{0.49\linewidth}
         \centering
         \includegraphics[width=\textwidth]{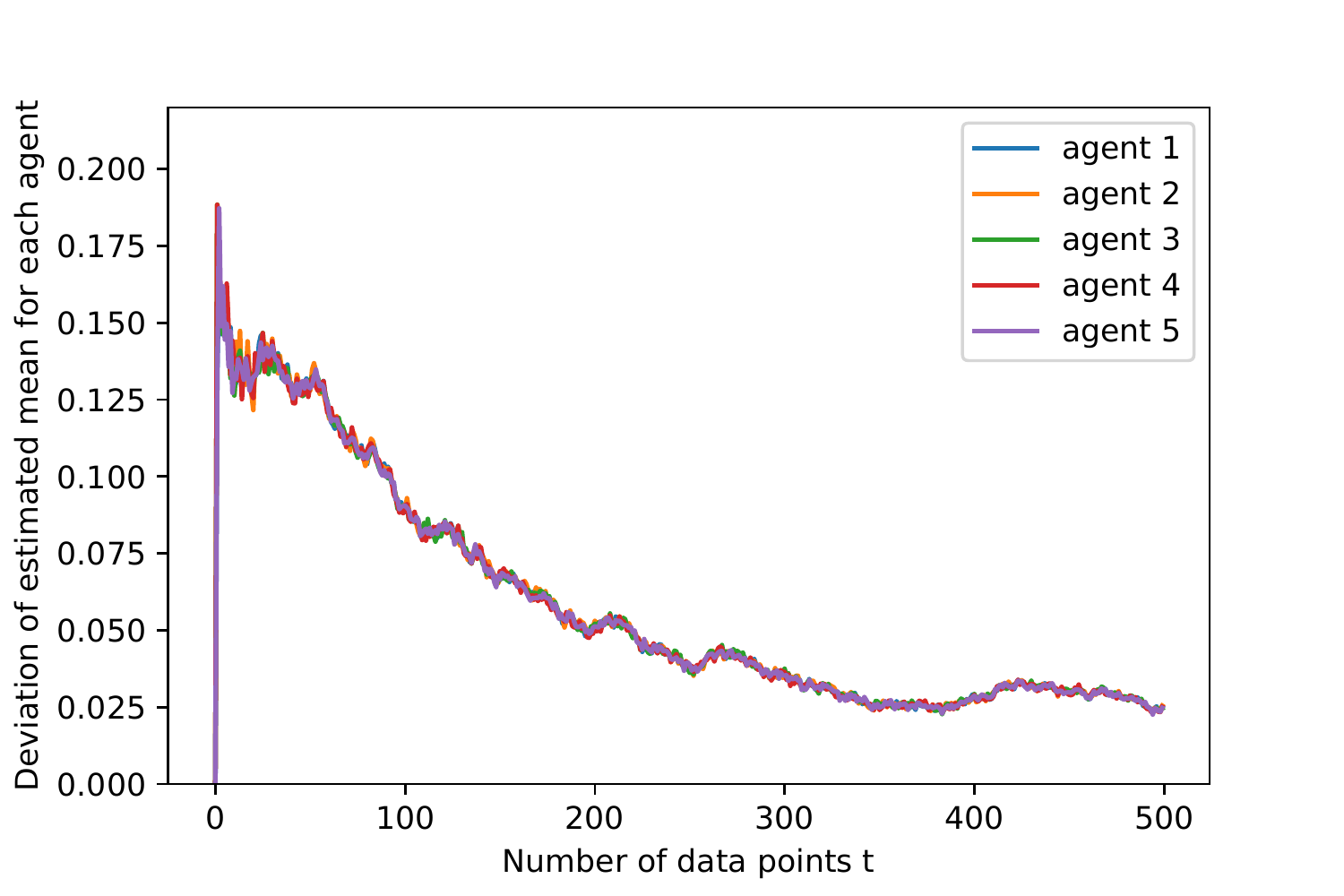}
         \caption{$K = 1$}
     \end{subfigure}
     \caption{Algorithm \ref{algo: online mean 2} with different parameter $K$.}
     \label{fig: dist simu 2}
\end{figure}

We next demonstrate the distributed algorithm using Algorithm \ref{algo: online mean 1} with a network of $m=5$ agents.
In Fig. \ref{fig: dist simu 1} and Fig. \ref{fig: dist simu 2}, we show the convergence of the absolute errors between the estimates and the true mean for each agent with different $K$ values (the number of communication rounds per data point), where $K = 0$ means no communication. All agents' estimates converge to the true mean with some error, and as $K$ increases, the estimates converge to the estimates of the centralized algorithm at a faster rate.

\section{Conclusions and Future Works}
In this work, we proposed online adaptations for robust mean estimation from arbitrarily corrupted data. We subsequently proposed distributed algorithms for agents to collaboratively estimate the mean. We analyzed the convergence properties of the proposed algorithms. However, we only considered a simple case when the corruption level of each agent is upper bounded by $\eta$. 
In the future, we will consider $\eta$-corrupted global data sets, in which case the data of some agents can be entirely corrupted. Additionally, we would like to consider the multi-dimensional mean estimation problem, with the presence of partially observing agents that are receiving data from a subset of the dimensions.

\bibliographystyle{unsrt}
\bibliography{references}

\begin{thebibliography}{10}

\bibitem{landgren2016distributed}
Peter Landgren, Vaibhav Srivastava, and Naomi~Ehrich Leonard.
\newblock Distributed cooperative decision-making in multiarmed bandits:
  Frequentist and bayesian algorithms.
\newblock In {\em 2016 IEEE 55th Conference on Decision and Control (CDC)},
  pages 167--172. IEEE, 2016.

\bibitem{DGAMA}
Tong Yao and Shreyas Sundaram.
\newblock Distributed estimation of sparse inverse covariances.
\newblock In {\em 2021 60th IEEE Conference on Decision and Control (CDC)},
  pages 4077--4082, 2021.

\bibitem{mitra2020new}
Aritra Mitra, John~A Richards, and Shreyas Sundaram.
\newblock A new approach to distributed hypothesis testing and non-bayesian
  learning: Improved learning rate and byzantine resilience.
\newblock {\em IEEE Transactions on Automatic Control}, 66(9):4084--4100, 2020.

\bibitem{yin2018byzantine}
Dong Yin, Yudong Chen, Ramchandran Kannan, and Peter Bartlett.
\newblock Byzantine-robust distributed learning: Towards optimal statistical
  rates.
\newblock In {\em International Conference on Machine Learning}, pages
  5650--5659. PMLR, 2018.

\bibitem{dubey2020cooperative}
Abhimanyu Dubey and Alex Pentland.
\newblock Cooperative multi-agent bandits with heavy tails.
\newblock In {\em International Conference on Machine Learning}, pages
  2730--2739. PMLR, 2020.

\bibitem{leblanc2013resilient}
Heath~J LeBlanc, Haotian Zhang, Xenofon Koutsoukos, and Shreyas Sundaram.
\newblock Resilient asymptotic consensus in robust networks.
\newblock {\em IEEE Journal on Selected Areas in Communications},
  31(4):766--781, 2013.

\bibitem{ALON1999137}
Noga Alon, Yossi Matias, and Mario Szegedy.
\newblock The space complexity of approximating the frequency moments.
\newblock {\em Journal of Computer and System Sciences}, 58(1):137--147, 1999.

\bibitem{tukey1963trim}
John~W Tukey and Donald~H McLaughlin.
\newblock Less vulnerable confidence and significance procedures for location
  based on a single sample: Trimming/winsorization 1.
\newblock {\em Sankhy{\=a}: The Indian Journal of Statistics, Series A}, pages
  331--352, 1963.

\bibitem{ronchetti2009robust}
Elvezio~M Ronchetti and Peter~J Huber.
\newblock {\em Robust statistics}.
\newblock John Wiley \& Sons, 2009.

\bibitem{lugosi2019survey}
G{\'a}bor Lugosi and Shahar Mendelson.
\newblock Mean estimation and regression under heavy-tailed distributions: A
  survey.
\newblock {\em Foundations of Computational Mathematics}, 19(5):1145--1190,
  2019.

\bibitem{tukey1960survey}
John~W Tukey.
\newblock A survey of sampling from contaminated distributions.
\newblock {\em Contributions to probability and statistics}, pages 448--485,
  1960.

\bibitem{huber1992robust}
Peter~J Huber.
\newblock Robust estimation of a location parameter.
\newblock In {\em Breakthroughs in statistics}, pages 492--518. Springer, 1992.

\bibitem{diakonikolas2020outlier}
Ilias Diakonikolas, Daniel~M Kane, and Ankit Pensia.
\newblock Outlier robust mean estimation with subgaussian rates via stability.
\newblock {\em Advances in Neural Information Processing Systems},
  33:1830--1840, 2020.

\bibitem{lugosi2021robust}
Gabor Lugosi and Shahar Mendelson.
\newblock Robust multivariate mean estimation: the optimality of trimmed mean.
\newblock {\em The Annals of Statistics}, 49(1):393--410, 2021.

\bibitem{olfati2007consensus}
Reza Olfati-Saber, J~Alex Fax, and Richard~M Murray.
\newblock Consensus and cooperation in networked multi-agent systems.
\newblock {\em Proceedings of the IEEE}, 95(1):215--233, 2007.

\bibitem{nedic2010constrained}
Angelia Nedic, Asuman Ozdaglar, and Pablo~A. Parrilo.
\newblock Constrained consensus and optimization in multi-agent networks.
\newblock {\em IEEE Transactions on Automatic Control}, 55(4):922--938, 2010.

\end{thebibliography}

\clearpage
\appendix

\begin{lemma}[\cite{nedic2010constrained}]\label{lemma: nedic}
Let $0<\zeta<1$ and let $\{\gamma_t\}$ be a positive scalar sequence. Assume that $\lim_{t\to\infty}\gamma_t = 0$. Then
\[\lim_{t\to \infty}\sum_{l=0}^{t} \zeta^{t-l}\gamma_l = 0.\]
\end{lemma}

\subsection{Proof of Theorem \ref{thm: online 1}}
First we state the following lemma in order to complete the proof.
\begin{lemma}\label{lemma: bern}(Bernstein's inequality)
Let $x_1, \ldots, x_t \sim X$ be i.i.d. random variables. If $X- \mathbb{E}[X] \leq a$ for a given $a >0$, then 
for any $\delta \in (0,1)$ and $t \in \mathbb{N}$, we have
\begin{multline}\label{eqn: bernstein}
\mathbb{P}\Big(\frac{1}{t}\sum_{i=1}^t x_i - \mathbb{E}[X] < \frac{a}{3t}\log(1/\delta)+\sqrt{\frac{2\sigma_X^2\log(1/\delta)}{t}} \Big) \\ \geq 1 - \delta.     
\end{multline}
\end{lemma}

To simplify notation, we omit the subscript $i$ denoting the agent in the proof. Using Lemma \ref{lemma: Q bounds} and following the procedures similar to those in \cite{lugosi2021robust}, we obtain that on event $A$, for $t > t_0$, 
\begin{align}\label{eqn: therome 1 bernstein}
&\frac{1}{t-t_0}\sum_{j = t_0+1}^t \phi_{\alpha_{t_0},\beta_{t_0}}(x_j) \nonumber\\
&\leq \frac{1}{t-t_0} \sum_{j = t_0+1}^t \phi_{\mu+Q_{2\epsilon_{t_0}(\overline{X})}, \mu+Q_{1-\epsilon_{t_0}/2}(\overline{X})}(x_j)\nonumber\nonumber\\
&= \mathbb{E}[\phi_{\mu+Q_{2\epsilon_{t_0}(\overline{X})}, \mu+Q_{1-\epsilon_{t_0}/2}(\overline{X})}(X)] \nonumber\\
& \qquad+ \frac{1}{t-t_0}\sum_{j=t_0+1}^t\Big(\phi_{\mu+Q_{2\epsilon_{t_0}(\overline{X})}, \mu+Q_{1-\epsilon_{t_0}/2}(\overline{X})}(x_j) \nonumber\\
&\qquad - \mathbb{E}[\phi_{\mu+Q_{2\epsilon_{t_0}(\overline{X})}, \mu+Q_{1-\epsilon_{t_0}/2}(\overline{X})}(X)]\Big).
\end{align}
In \cite{lugosi2021robust}, it was shown that 
\begin{multline}\label{eqn: E upper}
\mathbb{E}[\phi_{\mu+Q_{2\epsilon_{t_0}(\overline{X})}, \mu+Q_{1-\epsilon_{t_0}/2}(\overline{X})}(X)] \\
\leq \mu + \mathbb{E}[\overline{X}_{\mathbf{1}_{\overline{X} \geq Q_{1-\epsilon_{t_0}/2}(\overline{X})}}] 
 \leq \mu + E(\epsilon_{t_0},X),
\end{multline}
and similarly,
\begin{equation*}
    \mathbb{E}[\phi_{\mu+Q_{2\epsilon_{t_0}(\overline{X})}, \mu+Q_{1-\epsilon_{t_0}/2}(\overline{X})}(X)] 
    \geq \mu - E(4\epsilon_{t_0},X).
\end{equation*}
Based on its definition, $\phi_{\mu+Q_{2\epsilon_{t_0}(\overline{X})}, \mu+Q_{1-\epsilon_{t_0}/2}(\overline{X})}(X)$ is upper bounded by $Q_{1-\epsilon_{t_0}/2}(\overline{X}) + \mu$ with variance at most $\sigma_X^2$. We define $T = t-t_0$ and apply \eqref{eqn: E upper} as well as \eqref{eqn: bernstein} to \eqref{eqn: therome 1 bernstein}, where 
\begin{multline*}
\phi_{\mu+Q_{2\epsilon_{t_0}(\overline{X})}, \mu+Q_{1-\epsilon_{t_0}/2}(\overline{X})}(X)
\\-\mathbb{E}[\phi_{\mu+Q_{2\epsilon_{t_0}(\overline{X})}, \mu+Q_{1-\epsilon_{t_0}/2}(\overline{X})}(X)] 
\\\leq a = Q_{1-\epsilon_{t_0}/2}(\overline{X}) + E(4\epsilon_{t_0},X).
\end{multline*}
Using Lemma \ref{lemma: bern}, with probability at least $1-\delta/4$, we have
\begin{multline*}
\frac{1}{T}\sum_{j = t_0+1}^t \phi_{\alpha_{t_0},\beta_{t_0}}(x_j) 
    \leq \mu + E(\epsilon_{t_0},X) + \sigma_X \sqrt{\frac{2\log(4/\delta)}{T}} \\ + \frac{Q_{1-\epsilon_{t_0}/2}(\overline{X})\log(4/\delta)}{T}+ \frac{E(4\epsilon_{t_0},X)\log(4/\delta)}{T}.    
\end{multline*}

From \eqref{eqn: upper bound on Q} and \eqref{eqn: epsilon def}, we have that $\forall t_0 > 0$ and $\forall t \geq 2t_0$,
\begin{align*}
    \frac{1}{T}Q_{1-\epsilon_{t_0}/2}(\overline{X})&\log(4/\delta) \leq \frac{\sigma_X}{\sqrt{4\eta+\frac{6\log(4/\delta)}{t_0}}}\frac{\log(4/\delta)}{T} \\
    &\leq\sigma_X\sqrt{\frac{\log(4/\delta)}{6}}\frac{\sqrt{t_0}}{T} \leq\sigma_X\sqrt{\frac{\log(4/\delta)}{6T}}.
\end{align*}
Note that $E(4\epsilon_{t_0},X)\log(4/\delta)/T \leq E(4\epsilon_{t_0},X)$ by the assumption that $\delta \geq 4e^{-(t-t_0)}$, and $E(\epsilon_{t_0},X)\leq E(4\epsilon_{t_0},X)$ from \eqref{eqn: E}. 
As a result,
\begin{equation*}
\frac{1}{T}\sum_{j = t_0+1}^t \phi_{\alpha_{t_0},\beta_{t_0}}(x_j) \leq \mu + 2E(4\epsilon_{t_0},X) + 2\sigma_X \sqrt{\frac{\log(4/\delta)}{T}}.       
\end{equation*}

The proof is similar for the lower tail. We obtain that, on the event $A$, with probability at least $1-\delta/2$,
\begin{multline}\label{eqn: estimator}
  \left|\frac{1}{T}\sum_{j =t_0+1}^t \phi_{\alpha_{t_0},\beta_{t_0}}(x_j) - \mu \right| \leq 2E(4\epsilon_{t_0},X) \\+ 2\sigma_X \sqrt{\frac{\log(4/\delta)}{T}}.  
\end{multline}

Since there are at most $\eta t$ points where $\phi_{\alpha_{t_0}, \beta_{t_0}}(x_j) \neq \phi_{\alpha_{t_0}, \beta_{t_0}}(\tilde{x}_j)$ at time step $t$, and the gap is bounded 
\begin{multline*}
    \left|\phi_{\alpha_{t_0}, \beta_{t_0}}(x_j) - \phi_{\alpha_{t_0}, \beta_{t_0}}(\tilde{x}_j) \right| \\ \leq |Q_{\epsilon_{t_0}/2}(\overline{X})| + |Q_{1-\epsilon_{t_0}/2}(\overline{X})|,
\end{multline*}
it follows that since $\eta \leq \epsilon_{t_0}/8$, $t \geq 2t_0$, and $t / T \leq 2$,
\begin{align*}
\frac{1}{T}\Big|\sum_{j = t_0 + 1}^t &\phi_{\alpha_{t_0}, \beta_{t_0}}(x_j) - \sum_{j = t_0 + 1}^t \phi_{\alpha_{t_0}, \beta_{t_0}}(\tilde{x}_j)\Big| \\
&\quad \leq \eta\frac{t}{T}(|Q_{\epsilon_{t_0}/2}(\overline{X})| + |Q_{1-\epsilon_{t_0}/2}(\overline{X})|)\\
&\quad \leq \frac{\epsilon_{t_0}}{2}\max\{|Q_{\epsilon_{t_0}/2}(\overline{X})|,|Q_{1-\epsilon_{t_0}/2}(\overline{X})|\}
\end{align*}
and 
\begin{align*}
\frac{\epsilon_{t_0}}{2}Q_{1-\epsilon_{t_0}/2}(\overline{X}) 
&= \mathbb{P}(\overline{X}\geq Q_{1-\epsilon_{t_0}/2}(\overline{X}))Q_{1-\epsilon_{t_0}/2}(\overline{X})\\
&=\mathbb{E}[Q_{1-\epsilon_{t_0}/2}(\overline{X})_{\mathbf{1}_{\overline{X} \geq Q_{1-\epsilon_{t_0}/2}(\overline{X})}}] \\
&\leq \mathbb{E}[\overline{X}_{\mathbf{1}_{ \overline{X}\geq Q_{1-\epsilon_{t_0}/2}(\overline{X})}}],
\end{align*}
where the last inequality comes from the condition $\overline{X}\geq Q_{1-\epsilon_{t_0}/2}$ being satisfied.

Finally, on event $A$, 
\begin{equation}\label{eqn: corrupt}
   \frac{1}{T} \left|\sum_{j = t_0+1}^t \phi_{\alpha_{t_0},\beta_{t_0}}(x_j) - \sum_{j = t_0+1}^t \phi_{\alpha_{t_0},\beta_{t_0}}(\tilde{x}_j)\right| \leq E(\epsilon_{t_0},X).    
\end{equation}
Using triangle inequality for \eqref{eqn: estimator} and \eqref{eqn: corrupt}, we arrive at the result.

\subsection{Proof of Lemma \ref{lemma: consensus conv 1}}
Define $T = t - t_0$. Considering uncorrupted samples, from \eqref{eqn: consensus 1}, we have 
\begin{align*}
    \tilde{\mu}_{t} &= \frac{1}{mT} \sum_{i = 1}^m \sum_{j = t_0+1}^t \phi_{\alpha_{o,t_0},\beta_{o,t_0}}({x}_{i,j})\\
    &\leq \mathbb{E}[\phi_{\mu+Q_{2\epsilon_{t_0}(\overline{X})}, \mu+Q_{1-\epsilon_{t_0}/2}(\overline{X})}(X)] \nonumber\\
& \qquad+ \frac{1}{mT}\sum_{i = 1}^m\sum_{j=t_0+1}^t\Big(\phi_{\mu+Q_{2\epsilon_{t_0}(\overline{X})}, \mu+Q_{1-\epsilon_{t_0}/2}(\overline{X})}(x_{i,j}) \nonumber\\
&\qquad - \mathbb{E}[\phi_{\mu+Q_{2\epsilon_{t_0}(\overline{X})}, \mu+Q_{1-\epsilon_{t_0}/2}(\overline{X})}(X)]\Big).
\end{align*}
Using Bernstein's inequality following procedures in Theorem \ref{thm: online 1}, we obtain 
\begin{multline*}
  \left|\frac{1}{mT}\sum_{i = 1}^m\sum_{j =t_0+1}^t \phi_{\alpha_{o,t_0},\beta_{o,t_0}}(x_{i,j}) - \mu \right| 
  \\\leq 2E(4\epsilon_{t_0},X) + 2\sigma_X \sqrt{\frac{\log(4/\delta)}{mT}}.  
\end{multline*}

Since there at at most $m\eta t$ corrupted points where $\phi_{\alpha_{o,t_0}, \beta_{o,t_0}}(x_{i,j}) \neq \phi_{\alpha_{o,t_0}, \beta_{o,t_0}}(\tilde{x}_{i,j})$ at time step $t$,
\begin{align*}
\frac{1}{mT}\Big|\sum_{i = 1}^m\sum_{j = t_0 + 1}^t &\left(\phi_{\alpha_{o,t_0}, \beta_{o,t_0}}(x_{i,j}) -  \phi_{\alpha_{o,t_0}, \beta_{t_0}}(\tilde{x}_{i,j})\right)\Big| \\
&\quad \leq \frac{m\eta t}{mT}(|Q_{\epsilon_{t_0}/2}(\overline{X})| + |Q_{1-\epsilon_{t_0}/2}(\overline{X})|)\\
&\quad \leq E(\epsilon_{t_0},X).
\end{align*}
With triangle inequality, we arrive at the result.

\subsection{Proof of Lemma \ref{lemma: consensus conv 2}}
Let $\phi_t = \phi_{\alpha_{o,t},\beta_{o,t}}$ denote the updated trimming operator at time step $t$. 
We have shown in Theorem \ref{thm: online 2} that for a set of sample paths of measure 1, there exists a sample path dependent finite time $\Bar{t}$ such that no more bad events $B_t$ occur for $t \geq \bar{t}$. 

Rewriting the update \eqref{eqn: consensus 2} with finite time $\bar{t}$, the deviation between the estimated mean and the true mean can be expressed as
\begin{align*}\label{eqn: algo2 intermediate}
|\tilde{\mu}_t - \mu| &= |\frac{1}{mt}(\sum_{i = 1}^m\sum_{j = 1}^{\bar{t}} \phi_{j}(\tilde{x}_{i,j}) + \sum_{i = 1}^m\sum_{j = \bar{t}+1}^t \phi_{j}(\tilde{x}_{i,j})) - \mu| \\
& \leq \frac{\bar{t}}{t} |\tilde{\mu}_{\bar{t}} - \mu| + |\frac{1}{mt}\sum_{i = 1}^m\sum_{j = \bar{t}+1}^t \phi_{j}(\tilde{x}_{i,j})- \frac{m(t-\bar{t})}{mt}\mu |.
\end{align*}


By definition, $\phi_t(\tilde{x}_{i,t}) \leq \beta_t, \forall i \in \mathcal{V}$ and since \eqref{eqn: beta} holds for all $t\geq \bar{t}$, from \eqref{eqn: upper bound on Q} and \eqref{eqn: epsilon def}, we have
\begin{equation*}
\phi_t(\tilde{x}_{i,t}) \leq  \frac{\sigma_X}{\sqrt{4\eta + 6\frac{\log(4/\delta)}{t}}} + \mu, \forall i \in \mathcal{V}.
\end{equation*}
Let $U_t = \frac{\sigma_X}{\sqrt{4\eta + 6\frac{\log(4/\delta)}{t}}}$. We can provide an upper bound
\begin{equation*}
\left|\frac{1}{mt}\sum_{i = 1}^m\sum_{j = \bar{t}+1}^t \phi_{j}(\tilde{x}_{i,j})- \frac{m(t-\bar{t})}{mt}\mu \right| \leq \frac{1}{t}\sum_{j = \bar{t}+1}^t U_j.
\end{equation*} 

\subsection{Proof of Theorem \ref{thm: dist convergence 1}}
We split the error bound  $|\hat{\mu}_{i,t}^K - \mu|$ using triangle inequality. The error can be split into the error introduced by consensus and the error from the convergence of data. For any $t_0 > 0$,
\begin{align}\label{eqn: thm3 eqn1}
|\hat{\mu}_{i,t_0+1}^K - \mu| &\leq |\hat{\mu}_{i,t_0+1}^K - \Tilde{\mu}_{t_0+1} | + |\Tilde{\mu}_{t_0+1}  - \mu|\\
&\leq c\lambda^K |\hat{\mu}_{i,t_0+1}^0 - \tilde{\mu}_{t_0+1}| + |\Tilde{\mu}_{t_0+1}  - \mu| \label{eqn: thm3 eqn1b}.
\end{align}
Notice that from \eqref{eqn: MA update}, $\hat{\mu}_{i,t}^0 \leq \hat{\mu}_{i,t-1}^K + \frac{\phi_{t_0}(\tilde{x}_{i,t})}{t-t_0}$, where $\phi_{t_0}$ denotes $\phi_{\alpha_{o,t_0},\beta_{o,t_0}}$. 
Define $e_t = \tilde{\mu}_{t-1} - \tilde{\mu}_{t} + \frac{\phi_{t_0}(\tilde{x}_{i,t})}{t-t_0}$ as the difference introduced by new data points at each time step $t$. We can express \eqref{eqn: thm3 eqn1b} as 
\begin{align}\label{eqn: thm3 eqn2}
    |\hat{\mu}_{i,t_0+1}^K - \mu|  
    &\leq c\lambda^K |\hat{\mu}_{i,t_0}^K - \tilde{\mu}_{t_0} + e_{t_0+1}| + |\Tilde{\mu}_{t_0+1}  - \mu|, \nonumber\\
    &\leq c\lambda^K (|\hat{\mu}_{i,t_0}^K - \tilde{\mu}_{t_0}| + |e_{t_0+1}|) + |\Tilde{\mu}_{t_0+1}  - \mu|.
\end{align}
From \eqref{eqn: thm3 eqn1} and \eqref{eqn: thm3 eqn2}, we can observe, $\forall t\geq t_0+1$, 
\begin{equation*}
 |\hat{\mu}_{i,t+1}^K - \Tilde{\mu}_{t+1} | \leq  c\lambda^K (|\hat{\mu}_{i,t}^K - \tilde{\mu}_{t}| + |e_{t+1}|). 
\end{equation*}
Similarly for the next time step $t = t_0 + 2$, we have 
\begin{align*}
|\hat{\mu}_{i,t_0+2}^K - \mu| 
&\leq |\hat{\mu}_{i,t_0+2}^K - \Tilde{\mu}_{t_0+2} | + |\Tilde{\mu}_{t_0+2}  - \mu| \\   
&\leq c\lambda^K (|\hat{\mu}_{i,t_0+1}^K - \tilde{\mu}_{t_0+1}| + |e_{t_0+2}|) \\ &\quad + |\Tilde{\mu}_{t_0+2}  - \mu| \\
& \leq c\lambda^K \Big (c\lambda^K \left(|\hat{\mu}_{i,t_0}^K - \tilde{\mu}_{t_0}| + |e_{t_0+1}|\right)
+ \\ &\quad |e_{t_0+2}| \Big) + |\Tilde{\mu}_{t_0+2}  - \mu|\\
&=  (c\lambda^K)^2  \left(|\hat{\mu}_{i,t_0}^K - \tilde{\mu}_{t_0}| + |e_{t_0+1}|\right)
\\ &\quad + c\lambda^K |e_{t_0+2}| + |\Tilde{\mu}_{t_0+2}  - \mu|
\end{align*}
For all $t \geq t_0+1$, we obtain 
\begin{multline}\label{eqn: dist estimation bound}
|\hat{\mu}_{i,t}^K - \mu| \leq (c\lambda^{K})^{t-t_0}|\hat{\mu}_{i,t_0}^K - \tilde{\mu}_{t_0}| \\+ \sum_{j = t_0+1}^t (c\lambda^{K})^{t+1-j} |e_j| + |\tilde{\mu}_t - \mu| .    
\end{multline}
Next, we derive the upper bound of $|e_t|$. From \eqref{eqn: consensus 1}, we see that $\forall t \geq t_0 + 1$,
\begin{align}\label{eqn: e upper}
    |e_t| \leq \left|\frac{\sum_{i=1}^m \phi_{t_0}(\tilde{x}_{i,t})}{m(t-t_0)}\right| + \left|\frac{\phi_{t_0}(\tilde{x}_{i,t})}{t-t_0}\right|\nonumber\\
    \leq \left|\frac{2\max(|\alpha_{o,t_0}|,|\beta_{o,t_0}|)}{t-t_0}\right|.
\end{align}

Substituting \eqref{eqn: e upper} and \eqref{eqn: tilde mu bound 1} for the corresponding terms in \eqref{eqn: dist estimation bound}, 
we obtain the result.

\end{document}